\documentclass[a4paper,10pt]{article}
\usepackage[utf8]{inputenc}

\usepackage{graphicx}
\usepackage{pgf}
\usepackage{tikz}
\usepackage{amssymb}
\usepackage{amsmath}
\usepackage{amsthm}
\usepackage{thmtools, thm-restate}
\usetikzlibrary{arrows,automata,shapes,snakes}
\usepackage{verbatim} 
\usepackage{wrapfig}
\usepackage{color}

\theoremstyle{plain}

\newcommand{\gamea}{\mathcal{M}}
\newcommand{\gameb}{\mathcal{G}}
\newcommand{\tm}{\mathcal{T}}
\newcommand{\pa}{\textup{Verifier}}
\newcommand{\pb}{\textup{Falsifier}}

\newcommand{\agt}{\Pi}

\newcommand{\prop}{\textup{AP}}
\newcommand{\play}{\textup{Play}}
\newcommand{\hist}{\textup{Hist}}
\newcommand{\strat}{\textup{Strat}}
\newcommand{\last}{\textup{last}}

\newcommand{\dirleft}{\textup{Left}}
\newcommand{\dirright}{\textup{Right}}

\newcommand{\ints}{\mathbb{Z}}
\newcommand{\nats}{\mathbb{N}}

\newcommand{\atlsa}{\Phi}
\newcommand{\atlsb}{\Psi}
\newcommand{\atla}{\varphi}
\newcommand{\atlb}{\psi}

\newcommand{\atlx}{\textbf{X}}
\newcommand{\atlg}{\textbf{G}}
\newcommand{\atlu}{\textbf{U}}
\newcommand{\atlf}{\textbf{F}}

\newcommand{\LTL}{\textup{LTL}}
\newcommand{\CTL}{\textup{CTL}}
\newcommand{\CTLs}{\textup{CTL}^*}
\newcommand{\ATL}{\textup{ATL}}
\newcommand{\ATLs}{\textup{ATL}^*}
\newcommand{\qATL}{\textup{QATL}}
\newcommand{\qATLs}{\textup{QATL}^*}

\newcommand{\auta}{\mathcal{A}}
\newcommand{\autb}{\mathcal{B}}

\newcommand{\traces}{\textup{Traces}}
\newcommand{\pred}{\textup{Pre}}

\newcommand{\pspace}{\textsc{PSpace}}
\newcommand{\expspace}{\textsc{ExpSpace}}
\newcommand{\exptime}{\textsc{ExpTime}}
\newcommand{\twoexpspace}{\textsc{2ExpSpace}}
\newcommand{\twoexptime}{\textsc{2ExpTime}}

\newcommand{\threeexptime}{\textsc{3ExpTime}}

\newtheorem{definition}{Definition}
\newtheorem{theorem}[definition]{Theorem}
\newtheorem{lemma}[definition]{Lemma}
\newtheorem{corollary}[definition]{Corollary}

\newtheorem{proposition}[definition]{Proposition}

\makeatletter
\newsavebox{\@brx}
\newcommand{\llangle}[1][]{\savebox{\@brx}{\(\m@th{#1\langle}\)}%
  \mathopen{\copy\@brx\kern-0.5\wd\@brx\usebox{\@brx}}}
\newcommand{\rrangle}[1][]{\savebox{\@brx}{\(\m@th{#1\rangle}\)}%
  \mathclose{\copy\@brx\kern-0.5\wd\@brx\usebox{\@brx}}}

\newcommand{\coal}[1]{\llangle #1 \rrangle}

\usepackage{authblk}
\title{Model-checking Quantitative Alternating-time Temporal Logic on One-counter Game Models}
\author{Steen Vester}
\date{}
\affil{Technical University of Denmark, Kgs. Lyngby, Denmark}

\begin{document}

\maketitle

\begin{abstract}

We consider quantitative extensions of the alternating-time temporal logics $\ATL/\ATLs$ called quantitative alternating-time temporal logics ($\qATL/\qATLs$) in which the value of a counter can be compared to constants using equality, inequality and modulo constraints. We interpret these logics in one-counter game models which are infinite duration games played on finite control graphs where each transition can increase or decrease the value of an unbounded counter. That is, the state-space of these games are, generally, infinite. We consider the model-checking problem of the logics $\qATL$ and $\qATLs$ on one-counter game models with VASS semantics for which we develop algorithms and provide matching lower bounds. Our algorithms are based on reductions of the model-checking problems to model-checking games. This approach makes it quite simple for us to deal with extensions of the logical languages as well as the infinite state spaces. The framework generalizes on one hand qualitative problems such as $\ATL/\ATLs$ model-checking of finite-state systems, model-checking of the branching-time temporal logics $\CTL$ and $\CTLs$ on one-counter processes and the realizability problem of LTL specifications. On the other hand the model-checking problem for $\qATL/\qATLs$ generalizes quantitative problems such as the fixed-initial credit problem for energy games (in the case of $\qATL$) and energy parity games (in the case of $\qATLs$).  Our results are positive as we show that the generalizations are not too costly with respect to complexity. As a byproduct we obtain new results on the complexity of model-checking $\CTLs$ in one-counter processes and show that deciding the winner in one-counter games with $\LTL$ objectives is $\twoexpspace$-complete.

\end{abstract}

\section{Introduction}

The alternating-time temporal logics $\ATL$ and $\ATLs$ \cite{AHK02} are used to specify temporal properties of systems in which several entities interact. They generalize the widely applied linear-time temporal logic $\LTL$ \cite{Pnu77} and computation tree logics $\CTL$ \cite{CE81} and $\CTLs$ \cite{EH86} to a multi-agent setting. Indeed, it is possible to specify and reason about what different coalitions of agents can make sure to achieve. The model-checking problem for alternating-time temporal logics subsumes the realizability problem for $\LTL$ \cite{PR89a,PR89b} which is the problem of deciding whether there exists a program satisfying a given $\LTL$ specification no matter how the environment behaves. This is closely related to the synthesis problem which consists of generating a program meeting such a specification. Properties in these logics are inherently qualitative and the model-checking problem for alternating-time temporal logics has primarily been treated in finite-state systems. However, in \cite{BG13} extensions of $\ATL$ and $\ATLs$ to the quantitative alternating-time temporal logics $\qATL$ and $\qATLs$ have been introduced. The purpose is to make the languages capable of expressing quantitative properties of multi-agent scenarios as well as deal with infinite-state systems. These are represented using unbounded counters in addition to a finite set of control states. Naturally, this leads to undecidability in many cases since already deciding the winner in a reachability game on a two-dimensional vector addition system with states (VASS) can already simulate the halting problem of a two-counter machine \cite{BJK10}. In order to regain decidability we focus on the subproblem of a single unbounded counter. This is a significant restriction from the multi-dimensional case, but it still lets us express many interesting properties of infinite-state multi-agent systems. For instance, the model-checking problem includes problems such as energy games \cite{BFLMS08} and energy parity games \cite{CD12} in which a system respectively needs to keep an energy level positive and needs to keep an energy level positive while satisfying a parity condition. These are expressible in $\qATL$ and $\qATLs$ as $\coal{\texttt{Sys}} \atlg (r > 0)$ and $\coal{\texttt{Sys}} (\atlg (r > 0) \wedge \atla_{\textup{parity}})$ respectively where $r$ is used to denote the current value of the counter. It can be compared to constants using relations in $\{<, \leq, = , \geq, > \}$. $\atla_{\textup{parity}}$ is a parity condition expressed as an $\LTL$ formula. It is quite natural to model systems with a resource (e.g. battery level, time, money) using a counter where production and consumption correspond to increasing and decreasing the counter respectively.

Let us give another example of a $\qATL$ specification. Consider the game in Figure \ref{fig:vending} modelling the interaction between the controller of a vending machine and an environment. The environment controls the rectangular states and the controller controls the circular state. Initially, the environment can insert a coin or request coffee. Upon either input the controller can decrease or increase the balance, dispense coffee or release control to the environment again.

\begin{figure}[here]
\begin{center}
\begin{tikzpicture}

\tikzstyle{every node}=[ellipse, draw=black, fill=none,
                        inner sep=5pt, minimum width=15pt, minimum height=15pt]

\draw (1,5) node [rectangle] (s0)	{$\bullet$};

\draw (3,5) node [rectangle] (s10)	{Insert coin};
\draw (3,3.5) node [rectangle] (s11)	{Request coffee};

\draw (5,5) node (s2)	{};

\draw (8,6.5) node [rectangle] (s30)	{Decrease};
\draw (8,5) node [rectangle] (s31)	{Increase};
\draw (8,3.5) node [rectangle] (s32)	{Dispense};

\draw (5,6.5) node [rectangle] (s33)	{Release};

\path[->] (s0) edge node [left, draw=none] {} (s10);
\path[->] (s0) edge node [left, draw=none] {} (s11);
\path[->] (s10) edge (s2);
\path[->] (s11) edge (s2);

\path[->] (s2) edge [bend left = 10] node [above left, draw=none] {-1} (s30);
\path[->] (s2) edge [bend left = 10] node [above = -5, draw=none] {+1} (s31);
\path[->] (s2) edge [bend left = 10] (s32);
\path[->] (s2) edge (s33);

\path[->] (s30) edge (s2);
\path[->] (s31) edge (s2);
\path[->] (s32) edge (s2);
\path[->] (s33) edge [bend right] (s0);

\end{tikzpicture}
\end{center}
\caption{Model of interaction between a vending machine controller and an environment.}
\label{fig:vending}
\end{figure}
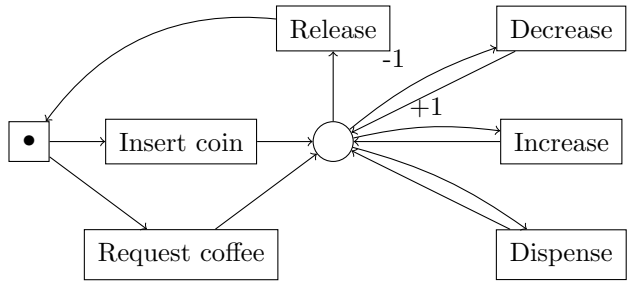

Some examples of specifications in $\qATLs$ using this model are

\begin{itemize}
 
 \item $\coal{\{\texttt{ctrl} \} } \atlg (\textup{Request} \wedge (r < 3) \rightarrow \atlx \atlx \textup{Release})$: The controller can make sure that control is released immediately whenever coffee is requested and the balance is less than 3.  
 
 \item $\coal{\{\texttt{ctrl} \} } \atlg (\textup{Request} \wedge (r \geq 3) \rightarrow \atlf \textup{Dispense})$: The controller can make sure that whenever coffee is requested and the balance is at least 3 then eventually a cup of coffee is dispensed.
\end{itemize}

\subsection{Contribution}

The contribution of this paper is to present algorithms and complexity results for model-checking $\qATL$ and $\qATLs$ in one-counter game models with one-dimensional VASS semantics, meaning that transitions that would make the counter go below zero are disabled. The complexity is investigated both in terms of whether only edge weights in $\{-1,0,+1\}$ can be used or if we allow any integer weights encoded in binary. We also distinguish between data complexity and combined complexity. In data complexity, the formula is assumed to be fixed whereas in combined complexity both the formula and the game are parameters. We characterize the complexity of the model-checking problems that arise from these distinctions for both $\qATL$ and $\qATLs$. In most of the cases the complexity results are quite satisfying compared with other results from the litterature. As a byproduct we also obtain precise data complexity for model-checking $\CTLs$ in one-counter processes (OCPs) and succinct one-counter processes (SOCPs). In addition, we show that the complexity of deciding the winner in a one-counter game with $\LTL$ objectives is $\twoexpspace$-complete. The complexity results encountered range from $\pspace$ to $\twoexpspace$, an overview of the results can be seen in Section \ref{sec:concluding}. The algorithms are based on model-checking games which makes it simple for us to handle the extensions of the logics considered as well as dealing with infinite state spaces and nesting of strategic operators.

\subsection{Related work}

The realizability problem for LTL was shown to be $\twoexptime$-complete in \cite{PR89a,PR89b}. As this problem is subsumed in $\qATLs$ model-checking this gives us an immediate $\twoexptime$ lower bound for the combined complexity of $\qATLs$ model-checking. The results for realizability of LTL specifications are generalized to quantitative objectives in \cite{BBFR13} where LTL objectives combined with a mean-payoff objective or an energy objective are considered. However, the semantics in their setting differs from ours in the way the counter value is handled when it gets close to 0. In our setting VASS semantics is used which is not the case in their setting. Our setting is equivalent to one-dimensional VASS games considered in e.g. \cite{BJK10}. Deciding the winner in games played on pushdown processes with parity objectives and LTL objectives were shown to be $\exptime$-complete and $\threeexptime$-complete in \cite{Wal01} and \cite{LMS04} respectively. Their setting is the same as ours except that in our setting a singleton stack alphabet is used to obtain one-counter games. In \cite{Ser06} it was shown that deciding the winner in one-counter parity games is in $\pspace$. It follows from \cite{BJK10} that this problem is $\pspace$-complete since selective zero-reachability in 1-dimensional VASS games is $\pspace$-hard. The approaches of module checking \cite{KVW01} and in particular pushdown module checking \cite{BMP05} are related to our setting and have given inspiration for our $\twoexpspace$-hardness proof of model-checking $\qATLs$. To compare, pushdown module checking of $\CTL$ and $\CTLs$ are $\twoexptime$-complete and $\threeexptime$-complete respectively. Our problems generalize several model-checking problems of branching-time logics in one-counter processes and are related to model-checking in pushdown processes as well. Model-checking of $\CTLs$ on pushdown processes has been shown decidable \cite{FWW97}, to be in $\twoexptime$ \cite{EKS01} and to be $\twoexptime$-hard \cite{Boz07}. On the other hand, model-checking $\CTL$ in succinct one-counter processes is $\expspace$-complete \cite{GHOW10}. Other related lines of research includes model-checking of Presburger LTL \cite{DG09} where counter constraints similar to (and more general than) ours are considered in the linear-time paradigm.

\section{Preliminaries}

A \emph{one-counter game} (OCG) is a particular kind of finitely representable infinite-state turn-based game. Such a game is represented by a finite game graph where each transition is labelled with an integer value from the set $\{-1,0,1\}$ as well as a counter that can hold any non-negative value. The idea is that when a transition labelled $v$ is taken when the counter value is $c$, the counter value changes to $c + v$. We require that transitions are only applicable when $c + v \ge 0$ since the counter cannot hold a negative value. When $r + v < 0$ we also say that the transition is disabled.

\begin{definition}

A one-counter game is a tuple $\gameb = (S, \agt, (S_j)_{j \in \agt}, R)$ where

\begin{itemize}

\item $S$ is a finite set of states

\item $\agt$ is a finite set of players

\item $S = \bigcup_{j \in \agt} S_j$ and $S_i \cap S_j = \emptyset$ for all $i,j \in \agt$ such that $i \neq j$

\item $R \subseteq S \times \{-1,0,1\} \times S$ is the transition relation

\end{itemize}

\end{definition}

An OCG is played by placing a token in an initial state $s_0$ and then moving the token between states for an infinite number of rounds. The transitions must respect the transition relation and the intuition is that for each $j \in \agt$ player $j$ controls the successor state when the token is placed on a state in $S_j$. At a given point in the game, the current counter value is given by the sum of the initial value $v_0 \in \nats$ and all the edge weights encountered so far. If a transition would make the current counter value decrease below 0 then the transition is disabled. More formally, an element $c \in S \times \nats$ is called a configuration of the game. We denote by $(S \times \nats)^*, (S \times \nats)^+$ and $(S \times \nats)^\omega$ the set of finite sequences, the set of non-empty finite sequences and the set of infinite sequences of configurations respectively. For a sequence $\rho = c_0 c_1 ...$ we define $\rho_i = c_i$, $\rho_{\leq i} = c_0 ... c_i$ and $\rho_{\ge i} = c_i c_{i+1} ...$. When $\rho$ is finite, i.e. $\rho = c_0 ... c_\ell$ we write $\last(\rho) = c_\ell$ and $|\rho| = \ell$. A play is a maximal sequence $\rho = (s_0,v_0) (s_1,v_1) ... $ of configurations such that for all $i \geq 0$ we have $(s_i,v_{i+1}-v_i,s_{i+1}) \in R$ and $v_i \geq 0$. A history is a proper prefix of a play. The set of plays and histories in an OCG $\gameb$ are denoted by $\play_\gameb$ and $\hist_\gameb$ respectively (the subscript may be omitted when it is clear from the context). The set of plays and histories with initial configuration $c_0$ are denoted $\play_\gameb(c_0)$ and $\hist_\gameb(c_0)$ respectively. A strategy for player $j \in \agt$ in $\gameb$ is a partial function $\sigma: \hist_\gameb \rightarrow S \times \nats$ defined for all histories $h = (s_0,v_0) ... (s_\ell, v_\ell) \in \hist_\gameb$ such that $s_\ell \in S_j$ with the requirement that if $\sigma(h) = (s,v)$ then $(s_\ell, v - v_\ell, s) \in R$. A play (resp. history) $\rho = c_0 c_1 ...$ (resp. $\rho = c_0 ... c_\ell$) is compatible with a strategy $\sigma_j$ for player $j \in \agt$ if $\sigma_j(\rho_{\leq i}) = \rho_{i+1}$ for all $i \ge 0$ (resp. $0 \leq i < \ell$) such that $\rho_i \in S_j \times \nats$. We denote by $\strat^j_\gameb$ the set of strategies of player $j$ in $\gameb$. For a coalition $A \subseteq \agt$ of players a collective strategy $\sigma = (\sigma_j)_{j \in A}$ is a tuple of strategies, one for each player in $A$. We denote by $\strat^A_\gameb$ the set of collective strategies of coalition $A$. For an initial configuration $c_0$ and collective strategy $\sigma = (\sigma_j)_{j \in A}$ of coalition $A$ we denote by $\play(c_0, \sigma)$ the set of plays with initial configuration $c_0$ that are compatible with $\sigma_j$ for every $j \in A$.

We extend one-counter games such that arbitrary integer weights are allowed and such that transitions are still disabled if they would make the counter go below zero. Such games are called succinct one-counter games (SOCGs). We suppose that weights are given in binary. The special cases of OCGs and SOCGs where $\agt$ is a singleton are called one-counter processes (OCPs) and succinct one-counter processes (SOCPs) respectively. A game model $\gamea = (\gameb, \prop, L)$ consists of a (one-counter or succinct one-counter) game $\gameb$, a finite set $\prop$ of atomic proposition symbols and a labelling $L: S \mapsto 2^{\prop}$ of the states $S$ of the game $\gameb$ with atomic propositions. We abbreviate one-counter game models and succinct one-counter game models by OCGM and SOCGM respectively.

By a one-counter parity game we mean the particular kind of one-counter game model where there are two players I and II and the set of propositions is a finite subset of the natural numbers, called colors. Further, every control state is labelled with exactly one color. In such a game, player I wins if the least color occuring infinitely often is even. Otherwise player II wins. We assume that the counter value is $0$ initially and that there is a designated initial state in a one-counter parity game. It was shown in \cite{Ser06} that the winner can be determined in a one-counter parity game in polynomial space by a reduction to the emptiness problem for alternating two-way parity automata \cite{Var98}.

\begin{proposition}
\label{prop:ocg_parity}

Determining the winner in one-counter parity games is in $\pspace$.

\end{proposition}

\section{Quantitative Alternating-time temporal logic}

We consider fragments of the quantitative alternating-time temporal logics $\qATL$ and $\qATLs$ introduced in \cite{BG13} interpreted over one-counter game models. The logics extend the standard $\ATL$ and $\ATLs$ \cite{AHK02} with atomic formulas of the form $r \bowtie c$ where $c \in \ints$ and $\bowtie \in \{\leq, <, =, >, \geq, \equiv_k \}$ with $k \in \nats$. They are interpreted in configurations of the game such that $r \leq 5$ is true if the current value of the counter is at most $5$ and $r = 0$ is true if the current value of the counter is $0$. $r \equiv_4 3$ means that the current value of the counter is equivalent to 3 modulo 4. More formally, the formulas of $\qATLs$ are defined with respect to a set $\prop$ of proposition symbols and a finite set $\agt$ of agents. They are constructed using the following grammar

$$\atlsa ::= p \mid r \bowtie c \mid \neg \atlsa_1 \mid \atlsa_1 \vee \atlsa_2 \mid \atlx \atlsa_1 \mid \atlsa_1 \atlu \atlsa_2 \mid \coal{A} \atlsa_1$$

where $p \in \prop$, $c \in \ints$, $\bowtie \in \{\leq, <, =, >, \geq, \equiv_k \}$ with $k \in \nats$, $A \subseteq \agt$ and $\atlsa_1, \atlsa_2$ are $\qATLs$ formulas. We define the syntactic fragment $\qATL$ of $\qATLs$ by the grammar

$$\atla ::= p \mid r \bowtie c \mid \neg \atla_1 \mid \atla_1 \vee \atla_2 \mid \coal{A} \atlx \atla_1 \mid \coal{A} \atlg \atla_1 \mid \coal{A} \atla_1 \atlu \atla_2 $$

where $p \in \prop$, $c \in \ints$, $\bowtie \in \{\leq, <, =, >, \geq, \equiv_k \}$ with $k \in \nats$, $A \subseteq \agt$ and $\atla_1, \atla_2$ are $\qATL$ formulas. Formulas of the form $r \bowtie c$ are called counter constraints.

We interpret formulas of $\qATL$ and $\qATLs$ in OCGMs. In standard $\ATLs$ we have state formulas and path formulas which are interpreted in states and plays respectively. For $\qATL$ and $\qATLs$ we also need the value of the counter to interpret state formulas. Note that the value of the counter is already present in a play. The semantics of a formula is defined with respect to a given OCGM $\gamea = (S, \agt, (S_j)_{j \in \agt}, R, \prop, L)$ inductively on the structure of the formula. For all states $s \in S$, plays $\rho \in \play_\gamea$, $p \in \prop$, $c,i \in \ints$, $A \subseteq \agt$, $\qATLs$ state formulas $\atlsa_1, \atlsa_2$ and $\qATLs$ path formulas $\atlsb_1, \atlsb_2$ let the satisfaction relation $\models$ be given by
\\

\begin{tabular}{l l l l }

$\gamea, s, i$ & $\models p$ & iff & $p \in L(s)$\\
$\gamea, s, i$ & $\models r \bowtie c$ & iff & $i \bowtie c$\\
$\gamea, s, i$ & $\models \neg \atlsa_1$ & iff & $\gamea, s, i \not\models \atlsa_1$\\
$\gamea, s, i$ & $\models \atlsa_1 \vee \atlsa_2$ & iff & $\gamea, s, i \models \atlsa_1$ or $\gamea, s, i \models \atlsa_2$\\
$\gamea, s, i$ & $\models \coal{A} \atlsb_1$ & iff & $\exists \sigma \in \strat^A_\gamea . \forall \pi \in \play_\gamea((s,i), \sigma) . \gamea, \pi \models \atlsb_1$\\
 & & & \\
$\gamea, \rho$ & $\models \atlsa_1$ & iff & $\gamea, \rho_0 \models \atlsa_1$\\
$\gamea, \rho$ & $\models \neg \atlsb_1$ & iff & $\gamea, \rho \not\models \atlsb_1$\\
$\gamea, \rho$ & $\models \atlsb_1 \vee \atlsb_2$ & iff & $\gamea, \rho \models \atlsb_1$ or $\gamea, \rho \models \atlsb_2$\\
$\gamea, \rho$ & $\models \atlx \atlsb_1$ & iff & $\gamea, \rho_{\ge 1} \models \atlsb_1$\\
$\gamea, \rho$ & $\models \atlsb_1 \atlu \atlsb_2$ & iff & $\exists k \ge 0 . \gamea, \rho_{\ge k} \models \atlsb_2$ and $\forall 0 \leq i < k . \gamea, \rho_{\ge i} \models \atlsb_1$ 

\end{tabular}
\\

The definition of the semantics is extended in the natural way to SOCGMs.

In this paper we focus on the model-checking problem. That is to decide, given an OCGM/SOCGM $\gamea$, a state $s$ in $\gamea$, a natural number $i$ and a $\qATL$/$\qATLs$ formula $\atla$ whether $\gamea, s, i \models \atla$. When doing model-checking we assume that states are only labelled with atomic propositions that occur in the formula $\atla$ as well as the special propositions $\top$ and $\bot$ that are true in all states and false in all states respectively. This is done to ensure that the input is finite. When measuring the complexity of the model-checking problem we distinguish between data complexity and combined complexity. For data complexity, the formula $\atla$ is assumed to be fixed and thus, the complexity only depends on the model. For combined complexity both the formula and game are assumed to be parameters. When model-checking OCGMs, the initial counter value $i$ is assumed to be input in unary and for SOCGMs, the initial counter value $i$ is assumed to be input in binary.

\section{Model-checking $\qATL$}

When model-checking $\ATL$ and $\ATLs$ in finite-state systems, the standard approach is to process the state subformulas from the innermost to the outermost, at each step labelling all states where the subformula is true. This approach does not work directly in our setting since we have an infinite number of configurations. We therefore take a different route and develop a model-checking game in which we can avoid explicitly labelling the configurations in which a subformula is true. This approach also allows us to handle the counter constraints in a natural way.

\subsection{A model-checking game for $\qATL$}

We convert the model-checking problem asking whether $\gamea, s_0, i \models \atla$ for a $\qATL$ formula $\atla$ in a configuration $(s_0,i)$ of an OCGM $\gamea = (S, \agt, (S_j)_{j \in \agt}, R, \prop, L)$ to a model-checking game $\gameb_{\gamea, s_0, i}(\atla)$ between two players $\pa$ and $\pb$ that are trying to respectively verify and falsify the formula. The construction is done so $\pa$ has a winning strategy in $\gameb_{\gamea, s_0, i}(\atla)$ if and only if $\gamea, s_0, i \models \atla$. The model-checking game can be constructed in polynomial time and is an OCG with a parity winning condition. According to Proposition \ref{prop:ocg_parity} determining the winner in such a game can be done in $\pspace$.

The construction is done inductively on the structure of $\atla$. For a given $\qATL$ formula, a given OCGM $\gamea$ and a given state $s$ in $\gamea$ we define a characteristic OCG $\gameb_{\gamea, s}(\atla)$. Note that the initial counter value is not present in the construction yet. There are a number of different cases to consider. We start with the base cases where $\atla$ is either a proposition $p$ or a formula of the form $r \bowtie c$ and then move on to the inductive cases. The circle states are controlled by $\pa$ and square states are controlled by $\pb$. $\pa$ wins the game if the least color that appears infinitely often during the play is even, otherwise $\pb$ wins the game. The states are labelled with colors whereas edges are labelled with counter updates.

$\gameb_{\gamea, s}(p): $ There are two cases. When $p \in L(s)$ and when $p \not\in L(s)$. The two resulting games are illustrated in Figure \ref{fig:casep} to the left and right respectively.

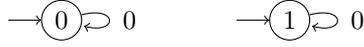
\begin{figure}[here]
\begin{center}
\begin{tikzpicture}

\tikzstyle{every node}=[ellipse, draw=black, fill=none,
                        inner sep=0pt, minimum width=15pt, minimum height=15pt]

\draw (3,6)     node (s0)	{0};

\draw (6,6)     node (s1)	{1};

\path[->] (s0) edge [loop right] node [right, draw=none] {0} (s0);
\path[->] (2.3,6) edge (s0);

\path[->] (s1) edge [loop right] node [right, draw=none] {0} (s1);
\path[->] (5.3,6) edge (s1);

\end{tikzpicture}
\end{center}
\caption{$\gameb_{\gamea,s}(p)$. To the left is the case where $p \in L(s)$ and to the right is the case where $p \not \in L(s)$}
\label{fig:casep}

\end{figure}
$\gameb_{\gamea, s}(r \bowtie c): $ Using negation and conjunction we can define $(r = c) \equiv (r \leq c \wedge \neg (r < c))$, $(r > c) \equiv (\neg (r \leq c))$ and $(r \ge c) \equiv (\neg (r < c))$ and therefore only need to construct games for the cases $r < c$, $r \leq c$ and $r \equiv_k c$. The three cases are shown in Figure \ref{fig:caserc}.

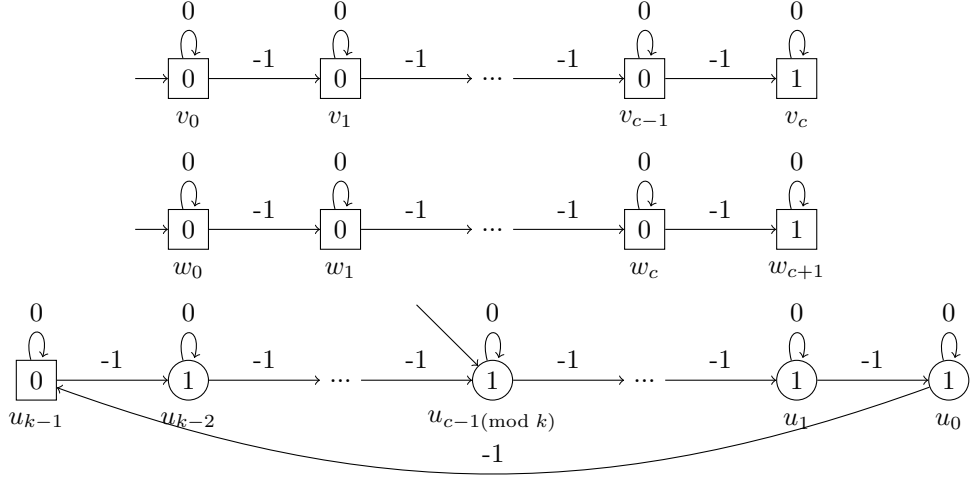
\begin{figure}[here]
\begin{center}
\begin{tikzpicture}

\tikzstyle{every node}=[ellipse, draw=black, fill=none,
                        inner sep=0pt, minimum width=15pt, minimum height=15pt]

\draw (2,8)     node [rectangle, label=below: $v_0$ ] (s00)	{0};
\draw (4,8)     node [rectangle, label=below: $v_1$ ] (s01)	{0};
\draw (6,8)     node [draw = none] (s02)	{...};
\draw (8,8)     node [rectangle, label=below: $v_{c-1}$ ] (s03)	{0};
\draw (10,8)     node [rectangle, label=below: $v_c$ ] (s04)	{1};

\draw (2,6)     node [rectangle, label=below: $w_0$ ] (s10)	{0};
\draw (4,6)     node [rectangle, label=below: $w_1$ ] (s11)	{0};
\draw (6,6)     node [draw = none] (s12)	{...};
\draw (8,6)     node [rectangle, label=below: $w_c$ ] (s13)	{0};
\draw (10,6)     node [rectangle, label=below: $w_{c+1}$ ] (s14)	{1};

\draw (0,4)     node [rectangle, label=below: $u_{k-1}$ ] (s20)	{0};
\draw (2,4)     node [label=below: $u_{k-2}$ ] (s21)	{1};
\draw (4,4)     node [draw = none] (s22)	{...};
\draw (6,4)     node [label=below: $u_{c-1 (\textup{mod } k)}$ ] (s23)	{1};
\draw (8,4)     node [draw = none] (s24)	{...};
\draw (10,4)     node [label=below: $u_{1}$ ] (s25)	{1};
\draw (12,4)     node [label=below: $u_{0}$ ] (s26)	{1};

\path[->] (1.3,8) edge (s00);
\path[->] (s00) edge [loop above] node [above, draw=none] {0} (s00);
\path[->] (s01) edge [loop above] node [above, draw=none] {0} (s01);
\path[->] (s03) edge [loop above] node [above, draw=none] {0} (s03);
\path[->] (s04) edge [loop above] node [above, draw=none] {0} (s04);

\path[->] (s00) edge node [above, draw=none] {-1} (s01);
\path[->] (s01) edge node [above, draw=none] {-1} (s02);
\path[->] (s02) edge node [above, draw=none] {-1} (s03);
\path[->] (s03) edge node [above, draw=none] {-1} (s04);

\path[->] (1.3,6) edge (s10);
\path[->] (s10) edge [loop above] node [above, draw=none] {0} (s10);
\path[->] (s11) edge [loop above] node [above, draw=none] {0} (s11);
\path[->] (s13) edge [loop above] node [above, draw=none] {0} (s13);
\path[->] (s14) edge [loop above] node [above, draw=none] {0} (s14);

\path[->] (s10) edge node [above, draw=none] {-1} (s11);
\path[->] (s11) edge node [above, draw=none] {-1} (s12);
\path[->] (s12) edge node [above, draw=none] {-1} (s13);
\path[->] (s13) edge node [above, draw=none] {-1} (s14);

\path[->] (5,5) edge (s23);
\path[->] (s20) edge [loop above] node [above, draw=none] {0} (s20);
\path[->] (s21) edge [loop above] node [above, draw=none] {0} (s21);
\path[->] (s23) edge [loop above] node [above, draw=none] {0} (s23);
\path[->] (s25) edge [loop above] node [above, draw=none] {0} (s25);
\path[->] (s26) edge [loop above] node [above, draw=none] {0} (s26);

\path[->] (s20) edge node [above, draw=none] {-1} (s21);
\path[->] (s21) edge node [above, draw=none] {-1} (s22);
\path[->] (s22) edge node [above, draw=none] {-1} (s23);
\path[->] (s23) edge node [above, draw=none] {-1} (s24);
\path[->] (s24) edge node [above, draw=none] {-1} (s25);
\path[->] (s25) edge node [above, draw=none] {-1} (s26);
\path[->] (s26) edge [bend left = 20] node [above, draw=none] {-1} (s20);

\end{tikzpicture}
\end{center}
\caption{$\gameb_{\gamea,s}(r < c)$ on top, $\gameb_{\gamea,s}(r \leq c)$ in the middle and $\gameb_{\gamea, s}(r \equiv_k c)$ at the bottom.}
\label{fig:caserc}

\end{figure}

$\gameb_{\gamea, s}(\atla_1 \vee \atla_2):$ The game is shown in Figure \ref{fig:casedis}.

\begin{figure}[here]
\begin{center}
\begin{tikzpicture}

\tikzstyle{every node}=[ellipse, draw=black, fill=none,
                        inner sep=0pt, minimum width=15pt, minimum height=15pt]

\draw (4,6)     node (s0)	{0};
\draw (5.5,6.66)     node [draw=none] (s1)	{$\gameb_{\gamea, s}(\atla_1)$};
\draw (5.5,5.33)     node [draw=none] (s2)	{$\gameb_{\gamea, s}(\atla_2)$};

\path[->] (3.3,6) edge (s0);
\path[->] (s0) edge node [above left, draw=none] {0} (s1);
\path[->] (s0) edge node [below left, draw=none] {0} (s2);

\end{tikzpicture}
\end{center}
\caption{$\gameb_{\gamea, s}(\atla_1 \vee \atla_2)$}
\label{fig:casedis}

\end{figure}

$\gameb_{\gamea, s}(\neg \atla_1):$ The game is constructed from $\gameb_{\gamea, s}(\atla_1)$ by interchanging circle states and square states and either adding or subtracting 1 to/from all colors.

$\gameb_{\gamea, s}(\coal{A} \atlx \atla_1):$ Let $R(s) = \{(s,v,s') \in R \} = \{(s,v_1,s_1),...,(s,v_m,s_m) \}$. There are two cases to consider. One when $s \in S_j$ for some $j \in A$ and one when $s \not \in S_j$ for all $j \in A$. Both are illustrated in Figure \ref{fig:casex}.

\begin{figure}[here]
\begin{center}
\begin{tikzpicture}

\tikzstyle{every node}=[ellipse, draw=black, fill=none,
                        inner sep=0pt, minimum width=15pt, minimum height=15pt]

\draw (3,6)     node (s0)	{0};
\draw (4.5,6.66)     node [draw=none] (s1)	{$\gameb_{\gamea, s_1}(\atla_1)$};
\draw (4.5,5.33)     node [draw=none] (s2)	{$\gameb_{\gamea, s_m}(\atla_1)$};
\draw (4,6)     node [draw=none] (s3)	{...};

\path[->] (2.3,6) edge (s0);
\path[->] (s0) edge node [above left, draw=none] {$v_1$} (s1);
\path[->] (s0) edge node [below left, draw=none] {$v_m$} (s2);

\draw (8,6)     node [rectangle] (t0)	{0};
\draw (9.5,6.66)     node [draw=none] (t1)	{$\gameb_{\gamea, s_1}(\atla_1)$};
\draw (9.5,5.33)     node [draw=none] (t2)	{$\gameb_{\gamea, s_m}(\atla_1)$};
\draw (9,6)     node [draw=none] (t3)	{...};

\path[->] (7.3,6) edge (t0);
\path[->] (t0) edge node [above left, draw=none] {$v_1$} (t1);
\path[->] (t0) edge node [below left, draw=none] {$v_m$} (t2);

\end{tikzpicture}
\end{center}
\caption{$\gameb_{\gamea, s}(\coal{A} \atlx \atla_1)$. The case on the left is when $s \in S_j$ for some $j \in A$ and the case on the right is when $s \not \in S_j$ for all $j \in A$}
\label{fig:casex}

\end{figure}

$\gameb_{\gamea, s}(\coal{A} \atlg \atla_1):$ In this case we let $\gameb_{\gamea, s}(\coal{A} \atlg \atla_1)$ have the same structure as $\gamea$, but with a few differences. $\pa$ controls all states that are in $S_j$ for some $j \in A$ and $\pb$ controls the other states. Further, for each transition $t = (s',v,s'') \in R$ we put an intermediate state $s_t$ controlled by $\pb$ between $s'$ and $s''$. When the player controlling $s'$ chooses to take the transition $t$ the play is taken to the intermediate state $s_t$ from which $\pb$ can either choose to continue to $s''$ or to go to $\gameb_{\gamea, s''}(\atla_1)$. Every state in $\gameb_{\gamea, s}(\coal{A} \atlg \atla_1)$ which is not part of $\gameb_{\gamea, s''}(\atla_1)$ has the color 0. It is illustrated in Figure \ref{fig:caseg}. Diamond states are states that can either be $\pa$ states or $\pb$ states. The intuition is that $\pb$ can challenge and claim that $\atla_1$ is not true in the current configuration. If he does so, $\pa$ must be able show that it is in fact true in order to win.

\begin{figure}[here]
\begin{center}
\begin{tikzpicture}

\tikzstyle{every node}=[ellipse, draw=black, fill=none,
                        inner sep=0pt, minimum width=15pt, minimum height=15pt]

\draw (1,6)     node [diamond,label=below:$s'$] (s0)	{};
\draw (3,6)     node [diamond,label=below:$s''$] (s1)	{};

\draw (6,6)     node [diamond,label=below:$s'$] (t0)	{0};
\draw (8,6)     node [rectangle, label=below right:$s_t$] (t1)	{0};
\draw (10,6)     node [diamond,label=below:$s''$] (t2)	{0};

\draw (8,5)     node [draw=none] (a)	{$\gameb_{\gamea, s''}(\atla_1)$};

\path[->] (s0) edge node [above, draw=none] {$v$} (s1);
\path[->] (t0) edge node [above, draw=none] {$v$} (t1);
\path[->] (t1) edge node [above, draw=none] {$0$} (t2);
\path[->] (t1) edge node [left, draw=none] {$0$} (a);

\draw[snake=snake] (4.1,6) -- (4.9,6);
\draw[-latex'] (4.9,6) -- (5.05,6);

\end{tikzpicture}
\end{center}
\caption{$\gameb_{\gamea, s}(\coal{A} \atlg \atla_1)$ is obtained by updating each transition in $\gamea$ as shown in the figure.}
\label{fig:caseg}

\end{figure}

$\gameb_{\gamea, s}(\coal{A} \atla_1 \atlu \atla_2):$ The game is constructed similarly to the case of $\coal{A} \atlg$. The differences are that every state is colored by 1 and for each transition $t = (s',v,s'') \in R$ we add two intermediate states $s_t$ and $s'_t$ controlled by $\pa$ and $\pb$ respectively with transitions to $\gameb_{\gamea, s''}(\atla_2)$ and $\gameb_{\gamea, s''}(\atla_1)$ respectively. The situation is illustrated in Figure \ref{fig:caseu}. The intuition is similar, but in this case $\pa$ loses unless he can claim $\atla_2$ is true at some point (and subsequently show that this is in fact the case). In addition $\atla_1$ cannot become false before this point, because then $\pb$ can claim that $\atla_1$ is false and win.

\begin{figure}[here]
\begin{center}
\begin{tikzpicture}

\tikzstyle{every node}=[ellipse, draw=black, fill=none,
                        inner sep=0pt, minimum width=15pt, minimum height=15pt]

\draw (0,6)     node [diamond,label=below:$s'$] (s0)	{};
\draw (2,6)     node [diamond,label=below:$s''$] (s1)	{};

\draw (5,6)     node [diamond,label=below:$s'$] (t0)	{1};
\draw (7,6)     node [label=below right:$s_t$] (t1)	{1};
\draw (9,6)     node [rectangle, label=below right:$s'_t$] (t2)	{1};
\draw (11,6)     node [diamond,label=below:$s''$] (t3)	{1};

\draw (7,5)     node [draw=none] (a)	{$\gameb_{\gamea, s''}(\atla_2)$};
\draw (9,5)     node [draw=none] (b)	{$\gameb_{\gamea, s''}(\atla_1)$};

\path[->] (s0) edge node [above, draw=none] {$v$} (s1);
\path[->] (t0) edge node [above, draw=none] {$v$} (t1);
\path[->] (t1) edge node [above, draw=none] {$0$} (t2);
\path[->] (t2) edge node [above, draw=none] {$0$} (t3);
\path[->] (t1) edge node [left, draw=none] {$0$} (a);
\path[->] (t2) edge node [left, draw=none] {$0$} (b);

\draw[snake=snake] (3.1,6) -- (3.9,6);
\draw[-latex'] (3.9,6) -- (4.05,6);

\end{tikzpicture}
\end{center}
\caption{$\gameb_{\gamea, s}(\coal{A} \atla_1 \atlu \atla_2)$ is obtained by updating each transition in $\gamea$ as shown in the figure.}
\label{fig:caseu}

\end{figure}

Finally, we define the game $\gameb_{\gamea, s, i}(\atla)$ from $\gameb_{\gamea, s}(\atla)$ and a natural number $i \in \nats$ as illustrated in Figure \ref{fig:casei}. Intuitively, this construction is performed to set the initial value of the counter to $i$.

It is now possible to prove the following result by induction on the structure of the $\qATL$ formula $\atla$, giving us a reduction from the model-checking problem to deciding the winner in a one-counter parity game.

\begin{restatable}{proposition}{reduc}
\label{prop:reduc}
For every OCGM $\gamea$, state $s$ in $\gamea$, $i \in \nats$ and $\atla \in \qATL$
$$\gamea, s, i \models \atla \textup{ if and only if } \pa \textup{ has a winning strategy in } \gameb_{\gamea,s,i}(\atla)$$

\end{restatable}

\begin{figure}[here]
\begin{center}
\begin{tikzpicture}

\tikzstyle{every node}=[ellipse, draw=black, fill=none,
                        inner sep=0pt, minimum width=15pt, minimum height=15pt]

\draw (2,6)     node [label=below: $v_0$ ] (s00)	{0};
\draw (4,6)     node [label=below: $v_1$ ] (s01)	{0};
\draw (6,6)     node [draw = none] (s02)	{...};
\draw (8,6)     node [label=below: $v_{i-1}$ ] (s03)	{0};
\draw (10,6)     node [draw = none] (s04)	{$\gameb_{\gamea, s}(\atla)$};

\path[->] (1.3,6) edge (s00);

\path[->] (s00) edge node [above, draw=none] {1} (s01);
\path[->] (s01) edge node [above, draw=none] {1} (s02);
\path[->] (s02) edge node [above, draw=none] {1} (s03);
\path[->] (s03) edge node [above, draw=none] {1} (s04);

\end{tikzpicture}
\end{center}
\caption{$\gameb_{\gamea,s,i}(\atla)$ is obtained by increasing the counter value to $i$ initially.}
\label{fig:casei}

\end{figure}
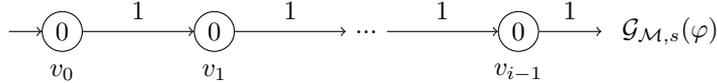

\subsection{Complexity}
\label{sec:lowerbound}

In \cite{BJK10} the selective zero-reachability problem for games on 1-dimensional vector addition systems with states was shown to be $\pspace$-complete. This problem consists of model-checking the fixed $\qATL$ formula $\coal{\{\textup{I}\}} \atlf (r = 0 \wedge p)$ in a 2-player OCGM where $\textup{I}$ is one of the players. The hardness is shown by a reduction from the emptiness problem of 1-letter alternating finite automata which is $\pspace$-complete \cite{JS07}. Thus, the data complexity of model-checking $\qATL$ in OCGMs is $\pspace$-hard. As a consequence of Proposition \ref{prop:reduc} and Proposition \ref{prop:ocg_parity} this lower bound is tight since we can transform the model-checking problem of $\qATL$ to deciding the winner in an OCG with a parity condition that has polynomial size. Thus, model-checking can be performed in polynomial space.



\begin{theorem}
\label{theo:qatl_pspacecomp}
The combined complexity and data complexity of model-checking $\qATL$ OCGMs are both $\pspace$-complete

\end{theorem}

In \cite{GHOW10} it was shown that the data complexity of model-checking $\CTL$ in SOCPs is $\expspace$-complete even for a fixed (but rather complicated) formula. Since this problem is subsumed by the model-checking problem of $\qATL$ in SOCGMs we have the same lower bound for the data complexity of model-checking $\qATL$ in SOCGMs. It can be shown that this bound is tight as follows. We can create a model-checking game for $\qATL$ in SOCGMs in the same way as for OCGMs and obtain a model-checking game which is an SOCG with a parity winning condition. This can be transformed into an OCG with a parity winning condition that is exponentially larger. It is done by replacing each transition with weight $v$ with a path that has $v$ transitions and adding small gadgets to make sure that a player loses if he tries to take a transition with value $-w$ for $w \in \nats$ when the current counter value is less than $w$. The exponential blowup is due to the weights being input in binary. We can then apply Proposition \ref{prop:ocg_parity} and solve this game in exponential space. Thus, we have the following.



\begin{theorem}

The combined complexity and data complexity of model-checking $\qATL$ in SOCGMs are both $\expspace$-complete.

\end{theorem}

These results are quite positive. Indeed, in OCGMs reachability games are already $\pspace$-complete \cite{BJK10}. Considering that in $\qATL$ we have nesting of strategic operators, eventuality operators, safety operators and comparison of counter values with constants it is very positive that we stay in the same complexity class. For SOCGMs $\CTL$ model-checking is already $\expspace$-complete \cite{GHOW10} which means that we can add several players as well as counter constraints without leaving $\expspace$.

\section{Model-checking $\qATLs$}

As for model-checking of $\qATL$ we rely on the approach of a model-checking game when model-checking $\qATLs$. However, due to the extended possibilities of nesting we do not handle temporal operators directly as we did for formulas of the form $\coal{A} \atla \atlu \atlb$, $\coal{A} \atlg \atla$ and $\coal{A} \atlx \atla$. Instead, we resort to a translation of $\LTL$ formulas into deterministic parity automata (DPA) which is combined with the model-checking game approach. This gives us model-checking games which are one-counter parity games as for $\qATL$, but with doubly exponential size in the input formula due to the translation from $\LTL$ formulas to DPAs.

\subsection{Adjusting the model-checking game to $\qATLs$}

Let $\gamea = (S, \agt, (S_j)_{j \in \agt}, R, AP, L)$ be an OCGM, $s_0 \in S$, $i \in \nats$ and $\atla$ be a $\qATLs$ state formula. The algorithm to decide whether $\gamea, s_0, i \models \atla$ follows along the same lines as our algorithm for $\qATL$. That is, we construct a model-checking game $\gameb_{\gamea, s_0, i}(\atla)$ between two players $\pa$ and $\pb$ that try to verify and falsify the formula respectively. Then $\pa$ has a winning strategy in $\gameb_{\gamea, s_0, i}(\atla)$ if and only if $\gamea, s_0, i \models \atla$. The construction is done inductively on the structure of $\atla$. For each state $s \in S$ and state formula $\atla$ we define a characteristic OCG $\gameb_{\gamea, s}(\atla)$. For formulas of the form $p, r \bowtie c, \neg \atla_1$ and $\atla_1 \vee \atla_2$ the construction is as for $\qATL$ assuming in the inductive cases that $\gameb_{\gamea, s}(\atla_1)$ and $\gameb_{\gamea, s}(\atla_2)$ have already been defined.

The interesting case is $\atla = \coal{A} \atla_1$. Here, let $\atlb_1,...,\atlb_m$ be the outermost proper state subformulas of $\atla_1$. Let $P = \{p_1,...,p_m\}$ be fresh propositions and let $f(\atla_1) = \atla_1[\atlb_1 \mapsto p_1,...,\atlb_m \mapsto p_m]$ be the formula obtained from $\atla_1$ by replacing the outermost proper state subformulas with the corresponding fresh propositions. Let $\prop' = \prop \cup P$. Now, $f(\atla_1)$ is an $\LTL$ formula over $\prop'$. We can therefore construct a deterministic parity automaton (DPA) $\auta_{f(\atla_1)}$ with input alphabet $2^{\prop'}$ such that the language $L(\auta_{f(\atla_1)})$ of the automaton is exactly the set of linear models of $f(\atla_1)$. The number of states of the DPA can be bounded by $O((2^{n\cdot2^n})(2^n)!)$ and the number of colors by $O(2 \cdot 2^n) = O(2^{n+1})$ where $n$ is the size of the formula $f(\atla_1)$. These bounds are obtained by using the fact that a non-deterministic B\" uchi automaton (NBA) $\autb_{f(\atla_1)}$ with $O(2^n)$ states and $L(\autb_{f(\atla_1)}) = \traces(f(\atla_1))$ can be constructed \cite{WVS83}. From this, a DPA accepting the same language can be constructed using a technique from \cite{Pit07} which translates an NBA with $m$ states to a DPA with $2m^m \cdot m!$ states and 2m colors.

The game $\gameb_{\gamea, s}(\atla)$ is now constructed with the same structure as $\gamea$, where $\pa$ controls the states for players in $A$ and $\pb$ controls the states for players in $\agt \setminus A$. However, we need to deal with truth values of the formulas $\atlb_1,...,\atlb_m$ which can in general not be labelled to states in $\gamea$ since they depend both on the current state and counter value. Therefore we change the structure to obtain $\gameb_{\gamea, s}(\atla)$ as follows. For each state $s$ and $t$ with $(s,t) \in R$ we embed a module as shown in Figure \ref{fig:updatest}. Here, $2^{\prop'} = \{\atlsa_0,...,\atlsa_\ell\}$ and for each $0 \leq j \leq \ell$ we let $\{\atlb_{j0},...,\atlb_{j k_j} \} = \{\atlb_i \mid p_i \in \atlsa_j \} \cup \{\neg \atlb_i \mid p_i \not \in \atlsa_j \}$.

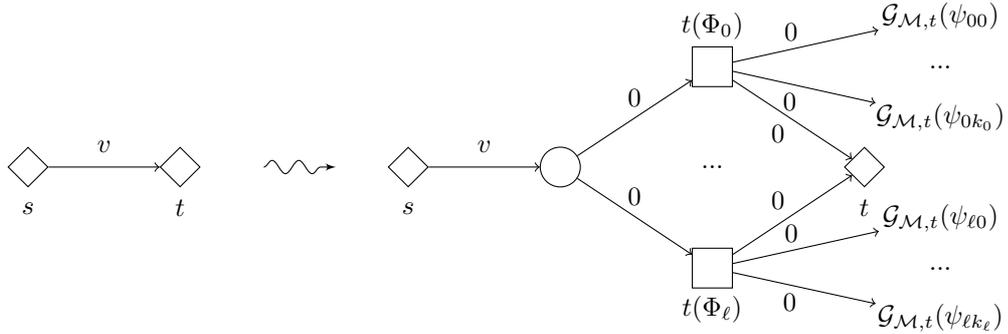
\begin{figure}[here]
\begin{center}
\begin{tikzpicture}

\tikzstyle{every node}=[ellipse, draw=black, fill=none,
                        inner sep=0pt, minimum width=15pt, minimum height=15pt]

\draw (0,7)     node [diamond,label=below:$s$] (s0)	{};
\draw (2,7)     node [diamond,label=below:$t$] (s1)	{};

\draw (5,7)     node [diamond,label=below:$s$] (t0)	{};
\draw (7,7)     node   (t1)	{};
\draw (9,8.33)     node [rectangle,label=above:$t(\atlsa_0)$] (t20)	{};
\draw (9,7)     node [draw=none] (t22)	{...};
\draw (9,5.66)     node [rectangle,label=below:$t(\atlsa_\ell)$] (t23)	{};

\draw (11,7)     node [diamond,label=below:$t$] (t3)	{};

\draw (12,9)     node [draw=none] (a)	{$\gameb_{\gamea, t}(\atlb_{0 0})$};
\draw (12,8.33)     node [draw=none] (b)	{...};
\draw (12,7.66)     node [draw=none] (c)	{$\gameb_{\gamea, t}(\atlb_{0 k_0})$};

\draw (12,6.33)     node [draw=none] (d)	{$\gameb_{\gamea, t}(\atlb_{\ell 0})$};
\draw (12,5.66)     node [draw=none] (e)	{...};
\draw (12,5)     node [draw=none] (f)	{$\gameb_{\gamea, t}(\atlb_{\ell k_\ell})$};

\path[->] (s0) edge node [above, draw=none] {$v$} (s1);
\path[->] (t0) edge node [above, draw=none] {$v$} (t1);
\path[->] (t1) edge node [above, draw=none] {$0$} (t20);
\path[->] (t1) edge node [above, draw=none] {$0$} (t23);
\path[->] (t20) edge node [below left, draw=none] {$0$} (t3);
\path[->] (t23) edge node [above left, draw=none] {$0$} (t3);
\path[->] (t20) edge node [above left, draw=none] {$0$} (a);
\path[->] (t20) edge node [below left, draw=none] {$0$} (c);
\path[->] (t23) edge node [above left, draw=none] {$0$} (d);
\path[->] (t23) edge node [below left, draw=none] {$0$} (f);

\draw[snake=snake] (3.1,7) -- (3.9,7);
\draw[-latex'] (3.9,7) -- (4.05,7);

\end{tikzpicture}
\end{center}
\caption{$\gameb_{\gamea, s}(\coal{A} \atla)$ is obtained by updating each transition as shown in the figure.}
\label{fig:updatest}

\end{figure}

The idea is that when a transition is taken from $(s,w)$ to $(t,w+v)$, $\pa$ must specify which of the propositions $p_1,...,p_m$ are true in $(t,w+v)$, this is done by picking one of the subsets $\atlsa_j$ (which is the set of propositions that are true in state $t(\atlsa_j)$). Then, to make sure that $\pa$ does not cheat, $\pb$ has the opportunity to challenge any of the truth values of the propositions specified by $\pa$. If $\pb$ challenges, the play never returns again. Thus, if $\pb$ challenges incorrectly, $\pa$ can make sure to win the game. However, if $\pb$ challenges correctly then $\pb$ can be sure to win the game. If $\pa$ has a winning strategy, then it consists in choosing the correct values of the propositions at each step. If $\pa$ does choose correctly and $\pb$ never challenges, the winner of the game should be determined based on whether the $\LTL$ property specified by $f(\atla_1)$ is satisfied during the play. We handle this by labelling $t(\atlsa_j)$ with the propositions in $\atlsa_j$. Further, since every step of the game is divided into three steps (the original step, the specification by $\pa$ and the challenge opportunity for $\pb$) we alter the deterministc automaton $\auta_{f(\atla_1)}$ such that it only takes a transition every third step. This simply increases its size by a factor 3. We then perform a product of the game with the updated parity automaton to obtain the parity game $\gameb_{\gamea, s}(\coal{A} \atla_1)$. It is important to note that the product with the automaton is not performed on the challenge modules (which are already colored), but only with states in the main module. This keeps the size of the game double-exponential in the size of the formula. We now have the following.

\begin{proposition}

For every OCGM $\gamea$, state $s$ in $\gamea$, $i \in \nats$ and state formula $\atla \in \qATLs$
$$\gamea, s, i \models \atla \textup{ if and only if } \pa \textup{ has a winning strategy in } \gameb_{\gamea, s, i}(\atla) $$

\end{proposition}

\begin{proof} Due to space limitations, we only provide a sketch of the proof with the main ideas. The proof is done by induction on the structure of $\atla$. The base cases as well as boolean combinations are omitted since they work as for $\qATL$. The interesting case is $\atla = \coal{A} \atla_1$.

Suppose first that $\gamea,s,i \models \coal{A} \atla_1$. Then coalition $A$ has a winning strategy $\sigma$ in $\gamea$. From this, we generate a strategy $\sigma'$ for $\pa$ in $\gameb_{\gamea, s, i}(\coal{A} \atla_1)$ that consists in never cheating when specifying values of atomic formulas and choosing transitions according to what $\sigma$ would have done in $\gamea$. Then, if $\pb$ challenges at some point, $\pa$ can be sure to win by the induction hypothesis since he never cheats. If $\pb$ never challenges (or, until he challenges), $\pa$ simply mimics the collective winning strategy $\sigma$ of coalition $A$ in $\gamea$ from $(s,i)$. This ensures that he wins in the parity game due to the definition of the parity condition from the parity automaton corresponding to $f(\atla_1)$.

Suppose on the other hand that $\pa$ has a winning strategy $\sigma$ in $\gameb_{\gamea, s, i}(\coal{A} \atla_1)$. Then $\sigma$ never cheats when specifying values of propositions, because then $\pb$ could win according to the induction hypothesis. Define a strategy $\sigma'$ for coalition $A$ in $\gamea$ that plays like $\sigma$ in the part of $\gameb_{\gamea, s, i}(\coal{A} \atla_1)$ where no challenge has occured. $\sigma'$ is winning for $A$ with condition $\atla_1$ in $\gamea$ due to the definition of $\gameb_{\gamea, s, i}(\coal{A} \atla_1)$ using the automaton $\auta_{f(\atla_1)}$.

\end{proof}

\subsection{Complexity}

The size of the model-checking game is doubly-exponential in the size of the formula. Therefore, it can be solved in doubly-exponential space because it is a one-counter parity game using Proposition \ref{prop:ocg_parity}. Actually, this is the case for both OCGMs and SOCGMs. Indeed, we extend the technique to SOCGMs as we did in the case of $\qATL$. However, with respect to complexity, the blowup caused by the binary representation of edge weights only matters when the formula is fixed since the game is already doubly-exponential when the input formula is a parameter. Thus, for $\qATLs$ we can do model-checking in doubly-exponential space whereas for a fixed formula it is in $\expspace$ for SOCGMs and $\pspace$ for OCGMs.

For combined complexity we can show that $\twoexpspace$ is a tight lower bound by a reduction from the word acceptance problem of a doubly-exponential space Turing machine. The reduction reuses ideas from \cite{JS07}, \cite{KMTV00} and \cite{BMP05}. The proof is in Appendix \ref{app:lower}. For a fixed formula we get tight lower bounds immediately from the results on $\qATL$.

\begin{restatable}{theorem}{twoexp}
\label{theo:twoexp}

The combined complexity of model-checking $\qATLs$ is $\twoexpspace$-complete for both OCGMs and SOCGMs. The data complexity of model-checking $\qATLs$ is $\pspace$-complete for OCGMs and $\expspace$-complete for SOCGMs.

\end{restatable}

Since we have an $\expspace$ lower bound for data complexity of $\CTL$ model-checking in SOCPs \cite{GHOW10} and a $\pspace$ lower bound for data complexity of $\CTL$ model-checking in OCPs \cite{GL13} we get the following results for data complexity of model-checking $\CTLs$ in OCPs.

\begin{corollary}
 
The data complexity of model-checking $\CTLs$ in OCPs and SOCPs are $\pspace$-complete and $\expspace$-complete respectively.
 
\end{corollary}

Since our lower bound is for formulas of the form $\coal{\{\textup{I}\}} \atla$ where $\atla$ is an $\LTL$ formula and $\textup{I}$ is a player we also have the following.

\begin{corollary}
 
 Deciding the winner in two-player OCGs and SOCGs with $\LTL$ objectives are both $\twoexpspace$-complete.
\end{corollary}

\section{Concluding remarks}
\label{sec:concluding}
We have characterized the complexity of the quantitative alternating-time temporal logics $\qATL$ and $\qATLs$ with respect to the format of edge weights as well as whether the input formula is fixed or not. The results are collected in Table \ref{tab:results}. Note that all complexity results on $\qATL$ and $\qATLs$ hold for $\ATL$ and $\ATLs$ as well since no counter constraints are used in the proofs of the lower bounds. As a byproduct we have also obtained results for $\CTLs$ model-checking on OCPs. These, along with $\CTL$ model-checking results on OCPs and SOCPs from the litterature, are included as a comparison.

\begin{table}[ht]
\caption{Complexity results of model-checking. Results for $\qATL$ and $\qATLs$ are on OCGMs and SOCGMs whereas results for $\CTL$ and $\CTLs$ are for OCPs and SOCPs respectively}
\begin{center}
\begin{tabular}{l | c c | c c}
 & \multicolumn{2}{c}{Non-succinct} & \multicolumn{2}{c}{Succinct}\\
\hline
 & Data & Combined & Data & Combined\\
\hline
$\qATL$ & $\pspace$-c & $\pspace$-c & $\expspace$-c & $\expspace$-c \\
$\qATLs$ & $\pspace$-c & $\twoexpspace$-c & $\expspace$-c & $\twoexpspace$-c \\
$\CTL$ & $\pspace$-c \cite{GL13} & $\pspace$-c \cite{GL13} & $\expspace$-c \cite{GHOW10} & $\expspace$-c \cite{GHOW10} \\
$\CTLs$ & $\pspace$-c & In $\twoexptime$ \cite{EKS01} & $\expspace$-c & In $\twoexpspace$
\end{tabular}
\end{center}
\label{tab:results}
\end{table}

Given that one-counter reachability games are already $\pspace$-complete \cite{BJK10} it is very positive that we can extend to $\qATL$ model-checking and even to model-checking of fixed $\qATLs$ formulas without leaving $\pspace$. Model-checking $\CTL$ in SOCPs is already $\expspace$-complete \cite{GHOW10} so it is also very positive that we can extend this to model-checking of $\qATL$ and fixed formulas of $\qATLs$ in succinct one-counter games. Finally, the $\twoexpspace$-completeness results are not too unexpteced compared to the known $\twoexptime$ lower bound from the synthesis of $\LTL$ \cite{PR89a} and $\threeexptime$-completeness of pushdown games with $\LTL$ objectives \cite{LMS04}. However, though we restrict to a unary stack alphabet compared to pushdown games, we do have counter constraints and nesting of strategic operators.

Finally, the model-checking game approach has turned out to be quite flexible with respect to enriching the alternating-time temporal logics with counter constraints. This is also the case when dealing with infinite state-spaces in which labelling of states with formulas that are true is not so straightforward. In addition, it has given us optimal complexity for most of the problems considered. We leave the combined complexity of $\CTLs$ model-checking open.

\subparagraph*{Acknowledgements}

I want to thank Valentin Goranko for discussions and helpful comments.

\appendix
\section{Full proof of Proposition \ref{prop:reduc}}

\reduc*

\begin{proof} 

The proof is done by induction on the structure of $\atla$. First, we consider the base cases.
\\

$\atla = p: $ In this case $\pa$ has a winning strategy if and only if $p \in L(s)$ if and only if $\gamea, s, i \models p$.
\\

$\atla = (r < c): $ In this case the counter is initially increased to $i$ after $i$ steps of the game. Then, $\pb$ can win exactly if he can decrease the counter $c-1$ times which is possible if and only if $c < i$. By the semantics of $\qATL$ this is exactly the case when $\gamea, s, i \models r < c$.
\\

$\atla = (r \leq c): $ The argument is similar to the case above.
\\

$\atla = (r \equiv_k c): $ In this case, $\pa$ has a winning strategy in $\gameb_{\gamea, s, i}(r \equiv_k c)$ if and only if he has a winning strategy where he subtracts one from the counter every time he can. The same is the case for $\pb$. For $\pa$ this is a winning strategy exactly when $\gamea, s, i \models (r \equiv_k c)$ if and only if $i \equiv_k c$. The reason is that after subtracting from the counter $i$ times, the current state will be $u_{k-1}$ if and only if
\\

$k-1 \equiv c-1 - i \textup{ (mod } k)$
\\

$\Leftrightarrow k \equiv c - i \textup{ (mod } k)$
\\

$\Leftrightarrow i \equiv c \textup{ (mod } k) \Leftrightarrow i \equiv_k c$
\\
\\
Next, we consider the inductive cases.
\\

$\atla = \atla_1 \vee \atla_2: $ Clearly, if $\pa$ has a winning strategy in $\gameb_{\gamea, s, i}(\atla_1)$ or in $\gameb_{\gamea, s, i}(\atla_2)$ then he has a winning strategy in $\gameb_{\gamea, s, i}(\atla_1 \vee \atla_2)$ since he can choose which of the games to play and reuse the winning strategy. On other hand, if $\pa$ has a winning strategy in $\gameb_{\gamea, s, i}(\atla_1 \vee \atla_2)$ then he is either winning in $\gameb_{\gamea, s, i}(\atla_1)$ or in $\gameb_{\gamea, s, i}(\atla_2)$ because he can reuse the strategy and be sure to win in at least one of these games. Then, by using the induction hypothesis we have that $\pa$ has a winning strategy in $\gameb_{\gamea, s, i}(\atla_1 \vee \atla_2)$ if and only if he has a winning strategy in $\gameb_{\gamea, s, i}(\atla_1)$ or in $\gameb_{\gamea, s, i}(\atla_2)$ if and only if $\gamea, s, i \models \atla_1$ or $\gamea, s, i \models \atla_2$ if and only if $\gamea, s, i \models \atla_1 \vee \atla_2$.
\\

$\atla = \neg \atla_1: $ The construction essentially switches $\pa$ with $\pb$ when creating $\gameb_{\gamea, s, i}(\neg \atla_1)$ from $\gameb_{\gamea, s, i}(\atla_1)$. This means that $\pa$ has a winning strategy in $\gameb_{\gamea, s, i}(\neg \atla_1)$ if and only if $\pb$ has a winning strategy in $\gameb_{\gamea, s, i}(\atla_1)$. As a consequence of the determinacy result for Borel games \cite{Mar75} we have that one-counter games with parity conditions are determined. It follows that $\pa$ has a winning strategy in $\gameb_{\gamea, s, i}(\neg \atla_1)$ if and only if $\pa$ does not have a winning strategy in $\gameb_{\gamea, s, i}(\atla_1)$. Using the induction hypothesis this means that $\pa$ has a winning strategy in $\gameb_{\gamea, s, i}(\neg \atla_1)$ if and only if $\gamea, s, i \not \models \atla_1$ if and only if $\gamea, s, i \models \neg \atla_1$.
\\

$\atla = \coal{A} \atlx \atla_1: $ There are two cases to consider. First, suppose $s \in S_j$ for some $j \in A$. Then $\pa$ has a winning strategy in $\gameb_{\gamea, s, i}(\coal{A} \atlx \atla_1)$ if and only if there is a transition $(s,v,s') \in R$ with $v+i \ge 0$ such that $\pa$ has a winning strategy in $\gameb_{\gamea, s', i+v}(\atla_1)$ since parity objectives are prefix independent. Using the induction hypothesis, this is the case if and only if there is a transition $(s,v,s') \in R$ with $v+i \ge 0$ such that $\gamea, s',i+v \models \atla_1$ which is the case if and only if $\gamea, s, i \models \coal{A} \atlx \atla_1$. For the case where $s \not \in S_j$ for all $j \in A$ the proof is similar, but uses universal quantification over the transitions.
\\

$\atla = \coal{A} \atlg \atla_1: $ The intuition of the construction is that $\pa$ controls the players in $A$ and $\pb$ controls the players in $\agt \setminus A$. At each configuration $(s',v) \in S \times \nats$ of the game $\pb$ can challenge the truth value of $\atla_1$ by going to $\gameb_{\gamea, s', v}(\atla_1)$ in which $\pb$ has a winning strategy if and only if $\atla_1$ is indeed false in $\gamea, s', v$. If $\pb$ challenges at the wrong time or never challenges then $\pa$ can make sure to win.

More precisely, suppose $\pa$ has a winning strategy $\sigma$ in $\gameb_{\gamea, s, i}(\coal{A} \atlg \atla_1)$ then every possible play when $\pa$ plays according to $\sigma$ either never goes into one of the modules $\gameb_{\gamea, s'}(\atla_1)$ or the play goes into one of the modules at some point and never returns. Since $\sigma$ is a winning strategy for I, we have by the induction hypothesis that every pair $(s',v) \in S \times \nats$ reachable when $\pa$ plays according to $\sigma$ is such that $\gamea, s', v \models \atla_1$, because otherwise $\sigma$ would not be a winning strategy for I. If coalition $A$ follows the same strategy $\sigma$ adapted to $\gamea$ then the same state, value pairs are reachable. Since for all these reachable pairs $(s',v)$ we have $\gamea, s', v \models \atla_1$ this strategy is a witness that $\gamea, s, i \models \coal{A} \atlg \atla_1$.

On the other hand, suppose that coalition $A$ can ensure $\atlg \atla_1$ from $(s,i)$ using strategy $\sigma$. Then in every reachable configuration $(s',v)$ we have $\gamea, s', v \models \atla_1$. From this we can generate a winning strategy for $\pa$ in $\gameb_{\gamea, s, i}(\coal{A} \atlg \atla_1)$ that plays in the same way until (if ever) $\pb$ challenges and takes a transition to a module $\gameb_{\gamea, s', v}(\atla_1)$ for some $(s',v)$. Since the same configurations can be reached before a challenge as when $A$ plays according to $\sigma$, this means that $\pa$ can make sure to win in $\gameb_{\gamea, s', v}(\atla_1)$ by the induction hypothesis. Thus, if $\pb$ challenges $\pa$ can make sure to win and if $\pb$ never challenges $\pa$ also wins since all states reached have color 0. Thus, $\pa$ has a winning strategy in $\gameb_{\gamea, s, i}(\coal{A} \atlg \atla_1)$.
\\

$\atla = \coal{A} \atla_1 \atlu \atla_2: $ The proof works as the case above with some minor differences. In this case, $\pa$ needs to show that he can reach a configuration where $\atla_2$ is true when controlling the players in $A$ and therefore he loses if he can never reach a module $\gameb_{\gamea, s', v}(\atla_2)$ such that $\gamea, s', v \models \atla_2$. At the same time, he has to make sure that configurations $(s',v)$ where $\gamea, s', v \not \models \atla_1$ are not reached in an intermediate configuration since $\pb$ still has the ability to challenge, as in the previous case. Note that $\pa$ gets the chance to commit to showing that $\atla_2$ is true in a given configuration before $\pb$ gets the change to challenge the value of $\atla_1$. This is due to the definition of the until operator that does not require $\atla_1$ to be true at the point where $\atla_2$ becomes true. We leave out the remaining details.

\end{proof}

\section{Full proof of Theorem \ref{theo:twoexp}}
 
\label{app:lower}
 
We will show that model-checking $\ATLs$ in OCGMs is $\twoexpspace$-hard by a reduction from the word acceptance problem for a deterministic doubly-exponential space Turing machine. From this, the theorem follows from the observations in the main text.

Let $\tm = (Q, q_0, \Sigma, \delta, q_F)$ be a deterministic Turing machine that uses at most $2^{2^{|w|^k}}$ tape cells on input $w$ where $k$ is a constant and $|w|$ is the number of symbols in $w$. Here, $Q$ is a finite set of control states, $q_0 \in Q$ is the initial control state.  $\Sigma = \{0,1,\#, a, r\}$ is the tape alphabet containing the blank symbol $\#$ and special symbols $a$ and $r$ such that $\tm$ accepts immediately if it reads $a$ and rejects immediately if it reads $r$, $\delta: Q \times \Sigma \rightarrow Q \times \Sigma \times \{\dirleft, \dirright \}$ is the transition function and $q_F \in Q$ is the accepting state. If $\delta(q,a) = (q',a',x)$ we write $\delta_1(q,a) = q', \delta_2(q,a) = a'$ and $\delta_3(q,a) = x$. Let $\Sigma_I = \Sigma \setminus \{\#\}$. Now, let $w = w_1 ... w_{|w|} \in \Sigma_I^*$ be an input word. From this we construct an OCGM $\gamea$, an initial state $s_0$ and a $\qATLs$ formula $\atlsa$ all with size polynomial in $n = |w|^k$ and $|\tm|$ such that $\tm$ accepts $w$ if and only if $\gamea, (s_0,0) \models \atlsa$.

We use an intermediate step in the reduction for simplicity of the arguments. This is done by considering an OCG $\gameb = (S', \{\pa,\pb\}, (S'_\pa, S'_\pb), R')$ with two players $\pa$ and $\pb$ and an initial state $s'_0$ such that $\pa$ can force the play to reach $s'_F$ if and only if $\tm$ accepts $w$. However, the size of the set $S'$ of states will be doubly-exponential in $n$. The idea of this construction resembles a reduction from the word acceptance problem for polynomial-space Turing machines to the emptiness problem for alternating finite automata with a singleton alphabet used in \cite{JS07}. Afterwards we will reduce this to model-checking of the $\ATLs$ formula $\atlsa$ in $\gamea$ where $|S|$ is polynomial in $n$. This reduction can be performed by considering a more involved formula. We will use a technique similar to those used in \cite{KMTV00} and \cite{BMP05} to simulate a $2^{n}$-bit counter by using $\LTL$ properties and alternation between the players. This is the main trick to keep the state-space of $\gamea$ small.

We start with some notation. We assume that $\tm$ uses the tape cells numbered $1,...,2^{2^n}$ and that the tape head points to position $1$ initially. In addition, suppose for ease of arguments that there are two extra tape cells numbered $0$ and $2^{2^n}+1$ such that $\tm$ immediately accepts if the tape head reaches cell $0$ or cell $2^{2^n} + 1$. That is, cell $0$ and $2^{2^n} + 1$ holds the symbol $a$ initially. Further, assume without loss of generality that if $\tm$ halts it always does so with the tape head pointing to cell $1$ that contains the symbol $a$. Since $\tm$ is deterministic it has a unique (finite or infinite) run on the word $w$ which is a sequence $C^w_0 C^w_1 ...$ of configurations. Let $\Delta = \Sigma \cup (Q \times \Sigma)$. Then each configuration $C^w_i$ is a sequence in $\Delta^{2^{2^n} + 2}$ containing exactly one element in $Q \times \Sigma$ which is used to specify the current control state and location of the tape head. For instance, the initial configuration $C^w_0$ is given by

$$C^w_0 = a (q_0, w_1) w_2 w_3 ... w_{|w|} \# \# .... \# a$$

We use $C^w_i(j)$ to denote the $j$th element of configuration $C^w_i$. For a given element $d \in \Delta$ we define the set $\pred(d)$ of predecessor triples of $d$ as
\\

\begin{tabular}{r l}

$\pred(d) =$ & $\{(d_1,d_2,d_3) \in \Sigma^3 \mid d_2 = d \}$\\
 & $\cup \{((q,b),d_2,d_3) \in (Q \times \Sigma) \times \Sigma^2 \mid d = (\delta_1(q,b), d_2) \textup{ and } \delta_3(q,b) = \dirright \}$\\
 & $\cup \{((q,b),d_2,d_3) \in (Q \times \Sigma) \times \Sigma^2 \mid d = d_2 \textup{ and } \delta_3(q,b) \neq \dirright \} $\\
 & $\cup \{(d_1,d_2,(q,b)) \in \Sigma^2 \times (Q \times \Sigma) \mid d = (\delta_1(q,b), d_2) \textup{ and } \delta_3(q,b) = \dirleft \}$\\
 & $\cup \{(d_1,d_2,(q,b)) \in \Sigma^2 \times (Q \times \Sigma) \mid d = d_2 \textup{ and } \delta_3(q,b) \neq \dirleft \} $\\
 & $\cup \{(d_1,(q,b),d_3) \in \Sigma \times (Q \times \Sigma) \times \Sigma \mid d = \delta_2(q,b) \}$
\end{tabular}
\\

The idea is that given the three elements $C^w_i(j-1), C^w_i(j)$ and $C^w_i(j+1)$ one can uniquely determine $C^w_{i+1}(j)$ according to the definition of a Turing machine. $\pred(d)$ is then the set of all triples $(d_1,d_2,d_3)$ such that it is possible to have $C^w_i(j-1) = d_1, C^w_i(j) = d_2, C^w_i(j+1) = d_3$ and $C^w_{i+1}(j) = d$.

We now define the OCG $\gameb = ((S', \{\pa,\pb\}, (S'_\pa, S'_\pb), R'))$ by

\begin{itemize}

\item $S' = (\{0,...,2^{2^n} + 1 \} \times (\Delta \cup \Delta^3)) \cup \{s'_0, s'_z, s'_r, s'_F\}$

\item $S'_\pa = (\{0,...,2^{2^n} + 1 \} \times \Delta) \cup \{s'_0\}$

\item $S'_\pb = (\{0,...,2^{2^n} + 1 \} \times \Delta^3) \cup \{s'_z, s'_r, s'_F\}$
 
\item $R'$ is the least relation such that

\begin{itemize}

\item $(s'_0, 1, s'_0) \in R'$

\item $(s'_0, 0, (1, (q_F, a))) \in R'$

\item $((j,d), 0, (j,(d_1,d_2,d_3))) \in R'$ for all $j \in \{1,...,2^{2^n}\}$ and all $(d_1,d_2,d_3) \in \pred(d)$

\item For $j \in \{0,2^{2^n}+1\}$ we have $((j,a), 0, s'_F) \in R'$ and $((j,d),0,s'_r) \in R'$ when $d \neq a$

\item $((j,d), 0, s'_z) \in R'$ for all $(j,d)$ such that $C^w_0(j) = d$.

\item $(s'_z,0,s'_F) \in R'$

\item $(s'_z,-1,s'_r) \in R'$

\item $((j,(d_1,d_2,d_3)), -1, (j-1,d_1)) \in R'$ for all $j \in \{1,...,2^{2^n}\}$ and all $d_1,d_2,d_3 \in \Delta$

\item $((j,(d_1,d_2,d_3)), -1, (j,d_2)) \in R'$ for all $j \in \{1,...,2^{2^n}\}$ and all $d_1,d_2,d_3 \in \Delta$

\item $((j,(d_1,d_2,d_3)), -1, (j+1,d_3)) \in R'$ for all $j \in \{1,...,2^{2^n}\}$ and all $d_1,d_2,d_3 \in \Delta$

\end{itemize}

\end{itemize}

The different types of transitions are shown in Figure \ref{fig:2explower_init}, \ref{fig:2explower_pa} and \ref{fig:2explower_pb}. The intuition is that $\pa$ tries to show that $\tm$ accepts $w$ and $\pb$ tries to prevent this. Initially, $\pa$ can increase the counter to any natural number, assume he chooses $v$. If $\tm$ accepts $w$ it does so in a final configuration with the tape head pointing at cell $1$ holding the symbol $a$ with the current control state $q_F$. The game is now played by moving backwards from the state $(1,(q_F,a))$ holding this information. $\pa$ can choose a predecessor triple that leads to $(1,(q_F,a))$. Player $\pb$ then chooses one of the elements of the triple, the counter is decreased by one and the play continues like this. Finally, if the counter is $0$ in a state $(j,d)$ such that $C^w_0(j) = d$ then $\pa$ can win by going to $s'_z$ from which $\pb$ can only go to $s'_F$. We will argue that $\pa$ can make sure that this happens if and only if $\tm$ accepts $w$ after performing $v$ steps.

\begin{figure}
\begin{center}
\begin{tikzpicture}

\tikzstyle{every node}=[ellipse, draw=black, fill=none,
                        inner sep=5pt, minimum width=15pt, minimum height=15pt]

\draw (2,10) node [label = below:${s'_0}$] (s0)	{};
\draw (5,10) node [label = below:${(1,(q_F,a))}$] (s1)	{};

\path[->] (s0) edge [loop above] node [above, draw=none] {+1} (s0);
\path[->] (s0) edge node [above, draw=none] {0} (s1);

\end{tikzpicture}
\end{center}
\caption{From the initial state, $\pa$ can increase the counter to any natural number before starting the game.}
\label{fig:2explower_init}

\end{figure}
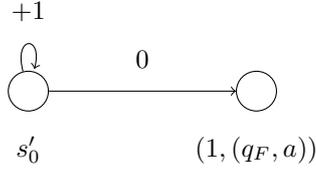

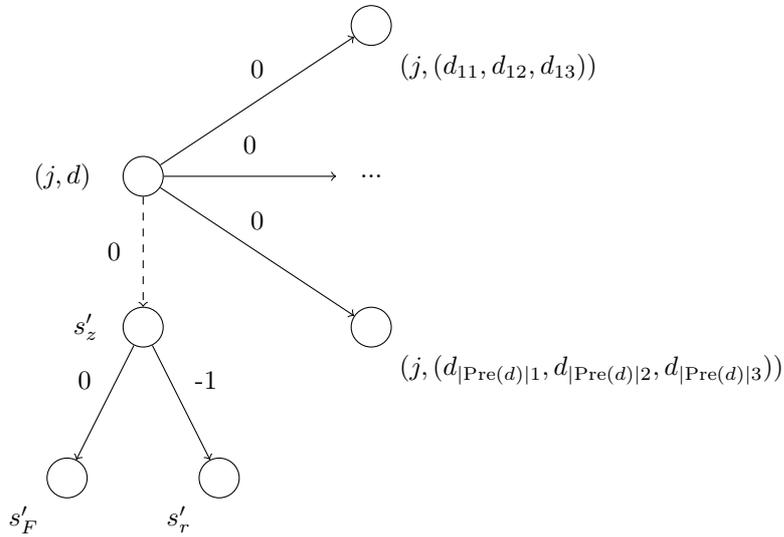
\begin{figure}
\begin{center}
\begin{tikzpicture}

\tikzstyle{every node}=[ellipse, draw=black, fill=none,
                        inner sep=5pt, minimum width=15pt, minimum height=15pt]

\draw (2,8) node [label = left: ${(j, d)}$] (s0)	{};

\draw (5,10) node [label = below right: ${(j,(d_{11},d_{12},d_{13}))}$] (s1)	{};
\draw (5,8) node [draw=none] (s2)	{...};
\draw (5,6) node [label = below right: ${(j,(d_{|\pred(d)|1},d_{|\pred(d)|2},d_{|\pred(d)|3}))}$] (s3)	{};

\draw (2,6) node [label = left: $s'_z$] (sz)	{};

\draw (1,4) node [label = below left: $s'_F$] (sf)	{};
\draw (3,4) node [label = below left: $s'_r$] (sr)	{};

\path[->] (s0) edge node [above, draw=none] {0} (s1);
\path[->] (s0) edge node [above, draw=none] {0} (s2);
\path[->] (s0) edge node [above, draw=none] {0} (s3);
\path[->] (s0) edge [dashed] node [left, draw=none] {0} (sz);
\path[->] (sz) edge node [above left, draw=none] {0} (sf);
\path[->] (sz) edge node [above right, draw=none] {-1} (sr);

\end{tikzpicture}
\end{center}
\caption{From a state $(j,d) \in \{1,...,2^{2^n} \} \times \Delta$ $\pa$ can choose a predecessor triple of $d$. The dashed transition is enabled only when $C^w_0(j) = d$. In this case $\pa$ can be sure to win if the current counter value is $0$.}
\label{fig:2explower_pa}

\end{figure}

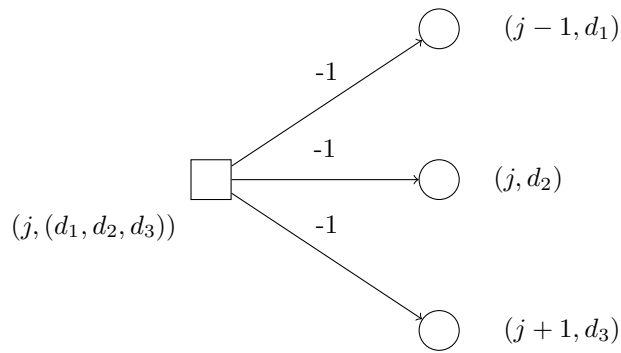
\begin{figure}
\begin{center}
\begin{tikzpicture}

\tikzstyle{every node}=[ellipse, draw=black, fill=none,
                        inner sep=5pt, minimum width=15pt, minimum height=15pt]

\draw (2,8) node [label = below left: ${(j, (d_1,d_2,d_3))}$, rectangle] (s0)	{};

\draw (5,10) node [label = right: ${(j-1,d_1)}$] (s1)	{};
\draw (5,8) node [label = right: ${(j,d_2)}$] (s2)	{};
\draw (5,6) node [label = right: ${(j+1,d_3)}$] (s3)	{};

\path[->] (s0) edge node [above, draw=none] {-1} (s1);
\path[->] (s0) edge node [above, draw=none] {-1} (s2);
\path[->] (s0) edge node [above, draw=none] {-1} (s3);

\end{tikzpicture}
\end{center}
\caption{From a precedessor triple chosen by $\pa$, $\pb$ can choose which predecessor to continue with.}
\label{fig:2explower_pb}

\end{figure}

\begin{lemma}
 The configuration $((j,d),i) \in (\{1,...,2^{2^n}\} \times \Delta) \times \nats$ is winning for $\pa$ if and only if $C^w_i(j) = d$. In particular $((1,(q_F,a)),i)$ is winning for $\pa$ if and only if $C^w_i(1) = (q_F,a)$ if and only if $\tm$ accepts $w$ after $i$ steps of computation.
 
\end{lemma}

\begin{proof}

The proof is done by induction on $i$. For the base case $i = 0$ the statement says that $((j,d),0)$ is winning for $\pa$ if and only if $C^w_0(j) = d$. Indeed, if $((j,d),0)$ is winning for $\pa$ he must go directly from $(j,d)$ to $s'_z$ because all other paths are blocked after one step since the counter value is $0$. If he goes to $s'_z$ then he wins because $\pb$ can only go to $s'_F$. However, note that there is only a transition from $(j,d)$ to $s'_z$ if $C^w_0(j) = d$ by construction. Thus, if $\pa$ is winning from $((j,d),0)$ then $C^w_0(j) = d$. For the other direction, suppose $C^w_0(j) = d$. Then $\pa$ can make sure to win by going to $s'_z$.

For the induction step, suppose the lemma is true for $i$. Now we need to show that $((j,d),i+1)$ is winning for $\pa$ if and only if $C^w_{i+1}(j) = d$. Suppose first that $((j,d),i+1)$ is winning for $\pa$. The winning strategy $\sigma$ cannot consist in going directly to $s'_z$ because then $\pb$ can go to $s'_r$. Thus, $\pa$ must choose a predecessor triple $(d_1,d_2,d_3) \in \pred(d)$ when playing according to $\sigma$. After he chooses this, $\pb$ chooses one of them and the counter is decreased by one. Thus, $\pb$ can choose either $((j-1,d_1),i)$, $((j,d_2),i)$ or $(j+1,d_3),i)$. Thus, by the induction hypothesis $C^w_i(j-1) = d_1$, $C^w_i(j) = d_2$ and $C^w_i(j+1) = d_3$ since $\pa$ is winning. By the definition of predecessor triples, this means that $C^w_{i+1}(j) = d$. For the other direction, suppose $C^w_{i+1}(j) = d$. Then by going to the state $(j,(C^w_i(j-1), C^w_i(j), C^w_i(j+1)))$ he can be sure to win by the induction hypothesis.

\end{proof}

\begin{lemma}
 Starting in configuration $(s'_0,0)$ $\pa$ can make sure to reach $s'_F$ if and only if $\tm$ accepts $w$.
\end{lemma}

We have now reduced the word acceptance problem to a reachability game in an OCG $\gameb$ with a doubly-exponential number of states. Due to the structure of $\gameb$ we can reduce this to model-checking the $\ATLs$ formula $\atlsa$ in the OCGM $\gamea$. The difficult part is that we need to store the number of the tape cell that the tape head is pointing at, which can be of doubly-exponential size. The other features of $\gameb$ are polynomial in the input. Note that at each step of the game, the position of the tape head either stays the same, increases by one or decreases by one. This is essential for our ability to encode it using $\ATLs$. We construct $\gamea$ much like $\gameb$ but where the position of the tape head is not present in the set of states. Instead, for each transition in the game between states $s$ and $s'$ we have a module in which $\pa$ encodes the position of the tape head by his choices. At the same time, $\pb$ has the possibility to challenge if $\pa$ has not chosen the correct value of the tape head position. This can be ensured by use of the $\ATLs$ formula $\atlsa = \coal{\{\pa\}} \atla$ where $\atla$ is an $\LTL$ formula. The details of simulating a $2^{n}$-bit counter like this can be obtained from \cite{KMTV00,BMP05}. According to the choices of $\pb$ then $\pa$ must be able to increase, decrease or leave unchanged the position of the tape head. This can be enforced by a formula with a size polynomial in $n$.  Except for having to implement the position of the tape head in this way, the rules of $\gamea$ are the same as for $\gameb$ where $\pa$ needs to show that $\tm$ accepts $w$ by choosing a strategy that ensures reaching a certain state in the game while updating the tape head position correctly. In the end, this means that for the initial state $s_0$ in $\gamea$ corresponding to $s'_0$ in $\gameb$ we get $\gamea, s_0, 0 \models \coal{\{\pa\}} (\atla \wedge \atlf s_{F})$ if and only if $\tm$ halts on $w$. Here we assume that the play also goes to a halting state $s_F$ corresponding to $s'_F$ if $\pb$ challenges the counter value incorrectly.

\twoexp*

\bibliographystyle{plain}
\bibliography{biblio}

\end{document}